\documentclass[11pt]{article}
\usepackage{fullpage}
\usepackage{amsmath}
\usepackage{amsthm}  
\usepackage{amssymb, algorithm, algpseudocode,eqparbox,thmtools,thm-restate,enumerate,verbatim,bm,array,tabu}
\usepackage[usenames,dvipsnames,svgnames,table]{xcolor}
\definecolor{darkgreen}{rgb}{0.0,0,0.9}
\usepackage[colorlinks=true,pdfpagemode=UseNone,citecolor=OliveGreen,linkcolor=BrickRed,urlcolor=BrickRed,pdfstartview=FitH,
pagebackref=true
]{hyperref}

\renewcommand{\P}[1]{{\mathbb{P}}\left[#1\right]}

\newcommand{\PP}[2]{{\mathbb{P}}_{#1}\left[#2\right]}

\newcommand{\E}{{\mathbb{E}}}
\newcommand{\iE}[1]{{\mathbb E}_{#1}}

\newcommand{\uE}[1]{\underset{#1}{\E}}

\newcommand{\st}{\mbox{\rm s.t. }}
\providecommand{\norm}[1]{\left\lVert#1\right\rVert}

\def\equationautorefname~#1\null{%
  equation~(#1)\null
}

\declaretheorem[numberwithin=section]{theorem}
\declaretheorem[sibling=theorem]{lemma}
\declaretheorem[sibling=theorem]{proposition}

\declaretheorem[sibling=theorem]{corollary}

\declaretheorem[sibling=theorem]{definition}

\declaretheorem[sibling=theorem]{problem}

\newenvironment{proofof}[1]{{\medbreak\noindent \em Proof of #1.  }}{\hfill\qed\medbreak}

\def\bone{{\bf 1}}

\def\eps{{\epsilon}}
\def\R{\mathbb{R}}
\def\C{\mathbb{C}}

\def\cF{{\cal F}}

\def\X{\chi}

\def\l{\ell}

\def\Ab{\textup{\bf Ab}}

\DeclareMathOperator{\polylog}{polylog}

\DeclareMathOperator{\trace}{Tr}

\DeclareMathOperator{\image}{Im}

\DeclareMathOperator{\conv}{conv}
\DeclareMathOperator{\interior}{relint}
\begin{document}

\title{The Kadison-Singer Problem for Strongly Rayleigh Measures \\ 
and  Applications to Asymmetric TSP}
\author{Nima Anari
\thanks{Computer Science Division, UC Berkeley.
Email: \protect\url{anari@berkeley.edu}.}
\and
Shayan Oveis Gharan
\thanks{Department of Computer Science and Engineering, University of Washington. This work was partly done while the author was a postdoctoral Miller fellow at UC Berkeley. Email: \protect\url{shayan@cs.washington.edu}.}
}

\date{}
\maketitle

\begin{abstract}

Marcus, Spielman, and Srivastava in their seminal work \cite{MSS13} resolved the Kadison-Singer conjecture by proving that for any set of finitely supported independently distributed random vectors $v_1,\dots, v_n$ which have ``small'' expected squared norm and are in isotropic position (in expectation), there is a positive probability that the sum $\sum v_i v_i^\intercal$ has small spectral norm. Their proof crucially employs real stability of polynomials which is the natural generalization of real-rootedness to multivariate polynomials.

Strongly Rayleigh distributions are families of probability distributions whose generating polynomials are real stable \cite{BBL09}. As independent distributions are just special cases of strongly Rayleigh measures, it is a natural question to see if the main theorem of \cite{MSS13} can be extended to families of vectors assigned to the elements of a strongly Rayleigh distribution.

In this paper we answer this question affirmatively; we show that for any homogeneous strongly Rayleigh distribution where the marginal probabilities are upper bounded by $\eps_1$  and any isotropic set of vectors assigned to the underlying elements whose norms are at most $\sqrt{\eps_2}$, there is a set in the support of the distribution such that  the spectral norm 
of the sum of the natural quadratic forms of the vectors assigned to the elements of the set is at most $O(\eps_1+\eps_2)$.
We employ our theorem to provide a sufficient condition for the existence of spectrally thin trees. This, together with a recent work of the authors \cite{AO14}, provides an improved upper bound on the integrality gap of the natural LP relaxation of the Asymmetric Traveling Salesman Problem. 
\end{abstract}

\section{Introduction}
Marcus, Spielman and Srivastava \cite{MSS13}
in a breakthrough work proved the Kadison-Singer conjecture \cite{KS59} by proving Weaver's \cite{Wea04} conjecture $\text{KS}_2$ and the Akemann and Anderson's Paving conjecture \cite{AA91}.
The following is their main technical contribution.
\begin{theorem}
\label{thm:mss}
If $\eps>0$ and $v_1,\ldots,v_m$ are independent random vectors in $\R^d$ with finite support where,
$$ \sum_{i=1}^m \E v_iv_i^\intercal = I,$$
such that for all $i$,
$$ \E\norm{v_i}^2 \leq \eps,$$
then
$$ \P{\norm{\sum_{i=1}^m v_iv_i^\intercal} \leq (1+\sqrt{\eps})^2} >0.$$
\end{theorem}
In this paper, we prove an extension of the above theorem to families of vectors assigned to elements of a not necessarily independent distribution.

Let $\mu:2^{[m]}\to R_+$ be a probability distribution on the subsets of the set $[m]=\{1,2,\ldots,m\}$. In particular, we assume that $\mu(.)$ is nonnegative and,
$$ \sum_{S\subseteq [m]} \mu(S) = 1.$$
 We assign a multi-affine polynomial with variables $z_1,\ldots,z_m$ to $\mu$,
 $$ g_\mu(z) = \sum_{S\subseteq [m]} \mu(S)\cdot z^S,$$
 where for a set $S\subseteq[m]$, $z^S=\prod_{i\in S} z_i$.
The polynomial $g_\mu$ is also  known as the {\em generating polynomial} of $\mu$. We say $\mu$ is a {\em homogeneous} probability distribution if $g_\mu$ is a homogeneous polynomial.

 We say that $\mu$ is a {\em strongly Rayleigh} distribution if $g_\mu$ is a real stable polynomial.
See \autoref{sec:realstability} for the definition of real stability.
Strongly Rayleigh measures are introduced and deeply studied in the seminal work of Borcea, Br\"and\'en and Liggett \cite{BBL09}.
They are natural generalizations of product distributions and cover several interesting families of probability distributions including determinantal measures and random spanning tree distributions.  
We refer interested readers to \cite{OSS11,PP14}
for applications of these probability measures.

Our main theorem extends \autoref{thm:mss} to
families of vectors assigned to the elements of a strongly Rayleigh distribution. This can be seen as a generalization because independent distributions are special classes of strongly Rayleigh measures. 
To state the main theorem we need another definition. 
The {\em marginal} probability of an element $i$ with respect to a probability distribution, $\mu$, is the probability that $i$ is in a sample of $\mu$,
\begin{equation}
\label{eq:marginaldef}
 \PP{S\sim\mu}{i\in S} = \partial_{z_i} g_\mu(z) \big|_{z_1=\ldots=z_m=1}. 
\end{equation}

\begin{theorem}[Main]
\label{thm:main}
Let $\mu$ be a homogeneous strongly Rayleigh probability distributions on $[m]$ such that the marginal probability of each element is at most $\eps_1$, and let $v_1,\ldots,v_m\in\R^d$ be vectors in an isotropic position,
$$ \sum_{i=1}^m v_iv_i^\intercal =I,$$
such that for all $i$, $\norm{v_i}^2 \leq \eps_2$.
Then, 
$$ \PP{S\sim \mu}{\norm{\sum_{i\in S} v_iv_i^\intercal} \leq 4(\eps_1+\eps_2)+2(\eps_1+\eps_2)^2 } >0.$$
\end{theorem}

The above theorem does not directly generalize \autoref{thm:mss}, but it can be seen as a variant of \autoref{thm:mss} to the case where the vectors $v_1,\dots,v_m$ are negatively dependent. 
We expect to see several applications of our main theorem that are not realizable by the original proof of \cite{MSS13}. 
In the following subsections we describe our main motivation for studying the above statement, which is to design approximation algorithms for the Asymmetric Traveling Salesman Problem (ATSP).

Let us conclude this part by proving a simple application of the above theorem to prove $\text{KS}_r$ for $r\geq 5$. 
\begin{corollary}\label{cor:KSr}
Given a set vectors $v_1,\dots,v_m\in\R^d$ 	in isotropic position,
$$ \sum_{i=1}^m v_iv_i^\intercal = I,$$
if for all $i$, $\norm{v_i}^2\leq \eps$, then for any $r$, there is an $r$ partitioning of $[m]$ into $S_1,\dots,S_r$ such that for any $j\leq r$,
$$ \norm{\sum_{i\in S_j} v_iv_i^\intercal} \leq 4(1/r+\eps)+2(1/4+\eps)^2.$$
\end{corollary}
\begin{proof} The proof is inspired by the lifting idea in \cite{MSS13}.
For $i\in [m]$ and $j\in [r]$ let $w_{i,j}\in\R^{d\cdot r}$ be the directed sum of $r$ vectors all of which are $0^d$ except the $j$-th one which is $v_i$, i.e.,
$$w_{i,1}=\begin{pmatrix}v_i\\ 0^d \\ \vdots \\ 0^d\end{pmatrix}, w_{i,2} = \begin{pmatrix}0^d \\ v_i \\ \vdots \\ 0^d\end{pmatrix}, \text{ and so on.} $$
Let $E=\{(i,j): i\in [m],j\in[r]\}$ and let $\mu:2^E\to\R_+$ be a product distribution defined in a way that selects exactly one pair $(i,j)\in E$ for any $i\in [m]$ uniformly at random. Observe that there are $m^r$ sets in the support of $\mu$ each of size exactly $m$ and each has probability $1/r^m$. Therefore, $\mu$ is a homogeneous probability distribution and the marginal probability of each element of $E$ is exactly $1/r$. In addition, since product distributions are strongly Rayleigh, $\mu$ is strongly Rayleigh.
Therefore, by \autoref{thm:main}, there is a set $S$ in the support of $\mu$ such that
$$ \norm{\sum_{(i,j)\in S} w_{i,j}w_{i,j}^\intercal} \leq \alpha,$$
for $\alpha=4(1/r+\eps) + 2(1/r+\eps)^2$.
Now, let $S_j=\{i: (i,j)\in S\}$. It follows that for any $j\in [r]$,
$$ \norm{\sum_{i\in S_j} v_iv_i^\intercal} \leq \alpha,$$
as desired.
\end{proof}

\subsection{The Thin Basis Problem}
In this section we use the main theorem to prove the existence of a thin basis among a given set of isotropic vectors. In the next section, we will use this theorem to prove the existence of  thin trees in graphs, i.e., trees which are ``sparse'' in all cuts of a given graph. 
 
Given a set of vectors $v_1,\dots,v_m\in \R^d$ in the isotropic position,
$$ \sum_{i=1}^m v_iv_i^\intercal = I,$$
we want to find a sufficient condition for the existence of a {\em thin basis}. Recall that a set $T\subset [m]$ is a basis if $|T|=d$ and all vectors indexed by $T$ are linearly independent. 
We say $T$ is $\alpha$-thin if
$$ \norm{\sum_{i\in T} v_iv_i^\intercal} \leq \alpha.$$
An obvious  necessary condition for the existence of an $\alpha$-thin basis is that the set
$$ V(\alpha):=\{v_i: \norm{v_i}^2 \leq \alpha\},$$
contains a basis. We show that there exist universal constants $C_1,C_2>0$ such that the existence of $C_1/\alpha$ disjoint bases in $V(C_2\cdot \alpha)$ is a sufficient condition.
\begin{theorem}\label{thm:thinbases}
	Given a set of vectors $v_1,\dots,v_m\in\R^d$ in the sub-isotropic position
	$$ \sum_{i=1}^m v_iv_i^\intercal \preceq I,$$
	if for all $1\leq i\leq m$, $\norm{v_i}^2\leq \eps$,
	and the set $\{v_1,\dots,v_m\}$ contains $k$ disjoint bases, then there exists an $O(\eps+1/k)$-thin basis $T\subseteq [m]$.
\end{theorem}

We will use \autoref{thm:main} to prove the above theorem. To use \autoref{thm:main} we need to define a strongly Rayleigh distribution on $[m]$ with small marginal probabilities.
This is proved in the following proposition.
\begin{proposition}\label{prop:determinantalconst}
Given a set of vectors $v_1,\dots,v_m\in\R^d$ that contains $k$ disjoint bases, there is a strongly Rayleigh probability distribution $\mu:2^{[m]}\to\R_+$ supported on the bases such that the marginal probability of each element is at most $O(1/k)$.
\end{proposition}
Now, \autoref{thm:thinbases} follows simply from the above proposition. Letting $\mu$ be  defined as above, we get $\eps_1=\eps$ and $\eps_2=O(1/k)$ in \autoref{thm:main} which implies the existence of a basis $T\subseteq [m]$ such that
$$ \norm{\sum_{i\in T} v_iv_i^\intercal} \leq O(\eps+1/k),$$
as desired.

In the rest of this section we prove the above proposition. 
In our proof $\mu$ will in fact be a homogeneous \emph{determinantal} probability distribution. We say $\mu:2^{[m]}\to\R_+$ is a determinantal probability distribution  
if there is a PSD matrix $M\in\R^{m\times m}$ such that
for any set $T\subseteq [m]$,
$$ \PP{S\sim\mu}{T\subseteq S} = \det(M_{T,T}), $$
where $M_{T,T}$ is the principal submatrix of $M$ whose rows and columns are indexed by $T$.
It is proved in \cite{BBL09} that any determinantal probability distribution is a strongly Rayleigh measure, so this is sufficient for our purpose.
In fact, we will find nonnegative weights $\lambda:[m]\to\R_+$ and for any basis $T$ we will let
\begin{equation} \mu_\lambda(T) \propto \det\left(\sum_{i\in T} \lambda_i v_iv_i^\intercal\right).\label{eq:mulambda}	
\end{equation}
It follows by the Cauchy-Binet identity that for any $\lambda$, such a distribution is determinantal with respect to the gram matrix $$M(i,j)=\sqrt{\lambda_i\lambda_j}\left\langle B^{-1/2} v_i,B^{-1/2} v_j\right\rangle$$ 
 where $B=\sum_{i=1}^m \lambda_i v_iv_i^\intercal$. 
So, all we need to do is find $\{\lambda_i\}_{1\leq i\leq m}$ such that the marginal probability of each element in $\mu_\lambda$ is $O(1/k)$. 

For any basis $T\subset [m]$ let $\bone_T\in \R^m$ be the indicator vector of the set $T$. Let $P$ be the convex hull of bases' indicator vectors,
$$ P:=\conv\{\bone_T: T\text{ is a basis}\}.$$
Recall that a point $x$ is in the \emph{relative interior} of $P$, $x\in \interior P$, if and only if $x$ can be written as a convex combination of all of the vertices of $P$ with strictly positive coefficients.

We find the weights in two steps. First, we show that for any point $x\in \interior P$, there exist weights $\lambda:[m]\to\R$ such that for any $i$,
$$ \PP{S\sim\mu_\lambda}{i\in S} = x(i),$$
where $x(i)$ is the $i$-th coordinate of $x$ and $\mu_\lambda$ is  defined as in \eqref{eq:mulambda}.
Then, we show that there exists a point $x\in \interior P$ such that for all $i$, $x(i) \leq O(1/k)$.
\begin{lemma}
	For any  $x\in\interior P$ there exist $\lambda:[m]\to\R_+$ such that the marginal probability of each element $i$ in $\mu_\lambda$ is $x(i)$.
\end{lemma}
\begin{proof}
Let $\mu^*:=\mu_\bone$ be the (determinantal) distribution where $\lambda_i=1$ for all $i$.
	The idea is to find a distribution $p(.)$ maximizing the
	relative entropy with respect to $\mu^*$ and preserves $x$ as the marginal probabilities. 
	This is analogous to the recent applications of maximum entropy distributions in approximation algorithms \cite{AGMOS10,SV14}.
	
Consider the following entropy maximization convex program.
\begin{equation}
\begin{aligned}
\min \hspace{5ex} & \sum_{T} p(T)\cdot \log \frac{p(T)}{\mu^*(T)} &\\
\st \hspace{2ex} & \sum_{T: i\in T} p(T) = x(i) & \forall i,\\
& p(T) \geq 0.
\end{aligned}	
\label{cp:maxentropy}
\end{equation}	
Note that  any feasible solution  satisfies  $\sum_T p(T)=1$ so we do not need to add this as a constraint.
First of all, since $x\in\interior P$, there exists a distribution $p(.)$ such that for all bases $T$, $p(T)>0$. So, the Slater condition holds and the duality gap of the above program is zero. 

Secondly, we use the Lagrange duality to characterize the optimum solution of the above convex program.
For any element $i$ let $\gamma_i$ be the Lagrange dual variable of the first constraint. The Lagrangian $L(p,\gamma)$ is defined as follows:
$$ L(p,\gamma) = \inf_{p\geq 0} \sum_T p(T)\cdot \log\frac{p(T)}{\mu^*(T)} - \sum_i \gamma_i \sum_{T:e\in T} (p(T)-x(i))$$
Let $p^*$ be the optimum $p$, letting the gradient of the RHS equal to zero we obtain, for any bases $T$,
$$ \log\frac{p^*(T)}{\mu^*(T)} + 1 = \sum_{i\in T} \gamma_i. $$
For all $i$, let $\lambda_i = \exp(\gamma_i - 1/d)$, where $d$ is the dimension of the $v_i$'s. Then, we get
\begin{eqnarray*} p^*(T) &=& \prod_{i\in T} \lambda_i \cdot \mu^*(T) \\
&=&	\prod_{i\in T} \lambda_i \cdot \det\left(\sum_{i \in T} v_iv_i^\intercal\right) \\
&=& \det\left(\sum_{i\in T} \lambda_i v_iv_i^\intercal\right).
\end{eqnarray*}
Therefore $p^*\equiv\mu_\lambda$.
Since the duality gap is zero, the above $p^*$ is indeed an optimal solution of the convex program~\eqref{cp:maxentropy}. Therefore, the marginal probability of every element $i$ with respect to $p^*$ ($\mu_\lambda$) is equal to $x(i)$.
\end{proof}

\begin{lemma}
	If $\{v_1,\dots, v_m\}$ contains $k$ disjoint bases, then there exists a point $x\in\interior P$, such that $x(i)=O(1/k)$ for all $i$.
\end{lemma}
\begin{proof}
	Let $T_1,\dots, T_k$ be the promised disjoint bases. Let $$x_0=\frac{\bone_{T_1}+\dots+\bone_{T_k}}{k}.$$
	The above is a convex combination of the vertices of $P$; so $x_0\in P$. We now perturb $x_0$ by a small amount to find a point in $\interior P$. Let $x_1$ be an arbitrary point in $\interior P$ (such as the average of all vertices). For any $\epsilon>0$, the point $x=(1-\epsilon)x_0+\epsilon x_1\in \interior P$. If $\epsilon$ is small enough, we get $x(i)=O(1/k)$ which proves the claim.
\end{proof}

This completes the proof of \autoref{prop:determinantalconst}.
\subsection{Spectrally Thin Trees}
For a graph $G=(V,E)$, the Laplacian of $G$, $L_G$, is defined as follows: For a vertex $i\in V$ let $\bone_i\in\R^V$ be the vector that is one at $i$ and zero everywhere else.
Fix an arbitrary orientation on the edges of $E$ and
let $b_e=\bone_i-\bone_j$ for an edge $e$ oriented from $i$ to $j$. Then,
$$ L_G=\sum_{e\in E} b_eb_e^\intercal.$$
We use $L_G^{\dagger}$ to denote the pseudo-inverse of $L_G$. Also, for a set $T\subseteq E$, we write
$$ L_T=\sum_{e\in T} b_eb_e^\intercal.$$
We say a spanning tree $T\subseteq E$ is $\alpha$-thin with respect to $G$ if for any set $S\subset V$,
$$ |T(S,\overline{S})|\leq \alpha \cdot |E(S,\overline{S})|,$$
where $T(S,\overline{S}),E(S,\overline{S})$ are the set of edges cross the cut $(S,\overline{S})$ in $T,G$ respectively. 
We say a spanning tree $T$ is $\alpha$-spectrally thin with respect to $G$ if 
$$ L_T \preceq \alpha\cdot L_G.$$
It is easy to see that spectral thinness is a generalization of the combinatorial thinness, i.e., if $T$ is $\alpha$-spectrally thin it is also $\alpha$-thin.

We say a graph $G$ is  $k$-edge connected if it has at least $k$ edges in any cut.
In recent works on Asymmetric TSP \cite{AGMOS10,OS11} it was shown that the existence of (combinatorially) thin trees in $k$-edge connected graphs
plays an important role in bounding the integrality gap of the natural linear programming relaxation of the Asymmetric TSP \cite{AO14}.

It turns out that the existence of spectrally thin trees is significantly easier to prove than combinatorially thin trees thanks to \autoref{thm:mss} of \cite{MSS13}.  
Given a graph $G=(V,E)$, Harvey and Olver \cite{HO14}  employ a recursive application of \cite{MSS13} and show that if for all edges $e\in E$, $b_e^\intercal L_G^{\dagger}b_e \leq \alpha$, then $G$ has an $O(\alpha)$-spectrally thin tree. The quantity $b_e L_G^{\dagger}b_e$ is  the effective resistance between the endpoints of $e$ when we replace every edge of $G$ with a resistor of resistance 1 \cite[Ch. 2]{LP13}. Unfortunately,
$k$-edge connectivity is a significantly weaker property than $\max_e b_e L_G^\dagger b_e\leq \alpha$ 
\cite{AO14}. So, this does not resolve the thin tree problem.

The main idea of \cite{AO14} is to slightly change the graph $G$ in order to decrease the effective resistance of edges while maintaining the size of the cuts intact. More specifically, to add a ``few'' edges $E'$ to $G$ such that in the new graph $G'=(V,E\cup E')$, the effective resistance of every edge of $E$ is small and the size of every cut of $G'$ is at most twice of that cut in $G$. If we can prove that $G'$ has a spectrally thin tree $T\subseteq E$ such a tree is combinatorially thin with respect to $G$ because $G,G'$ have the same cut structure. 
To show that $G'$ has a spectrally thin tree we need to answer the following question. 
\begin{problem}\label{pr:Feffres}
Given a graph $G=(V,E)$, suppose there is a set $F\subseteq E$ such that $(V,F)$ is $k$-edge connected, and that for all $e\in F$, $b_e^\intercal L_G^\dagger b_e \leq \alpha$.
Can we say that $G$ has  a $C\cdot \max\{\alpha,1/k\}$-spectrally thin tree for a universal constant $C$? 
\end{problem}

We use \autoref{thm:thinbases} to answer the above question affirmatively. 
Note that the above question cannot be answered by \autoref{thm:mss}. 
One can use \autoref{thm:mss} to show that the set $F$ can be partitioned  into two  sets $F_1,F_2$ such that each $F_i$ is $1/2+O(\alpha)$-spectrally thin, but \autoref{thm:mss} gives no guarantee on the connectivity of $F_i$'s. On the other hand, once we apply our main theorem to a strongly Rayleigh distribution supported on connected subgraphs of $G$, e.g. the spanning trees of $G$, we get connectivity for free.

\begin{corollary}
Given a graph $G=(V,E)$ and a set $F\subseteq E$ such that $(V,F)$ is $k$-edge connected, if for $\eps>0$ and any edge $e\in F$, $b_e^\intercal L_G^{\dagger} b_e \leq \eps$, then $G$ has an $O(1/k+\eps)$ spectrally thin tree.
\end{corollary}
\begin{proof}
Let $L_G^{\dagger/2}$ be the square root of $L_G^{\dagger}$. Note that since $L_G^\dagger \succeq 0$, its square root is well defined.
For all $e\in F$, let  
$$v_e = L_G^{\dagger/2} b_e.$$ 
Then, by the corollary's assumption, for each $e\in F$,
$$ \norm{v_e}^2 = b_e L_G^\dagger b_e \leq \eps,$$
and the vectors $\{v_e\}_{e\in F}$ are in sub-isotropic position,
\begin{eqnarray*}
 \sum_{e\in F} v_e v_e^\intercal &=& L_G^{\dagger/2} \left(\sum_{e\in F} v_e v_e^\intercal\right) L_G^{\dagger/2}\\
  &=& L_G^{\dagger/2} L_F L_G^{\dagger/2} \preceq I.
 \end{eqnarray*}
In addition, we can show that $\{v_e\}_{e\in F}$ contains $k/2$ disjoint bases. First of all, note that each basis of the vectors $\{v_e\}_{e\in F}$ corresponds to a spanning tree of the graph $(V,F)$. 
Nash-Williams \cite{NW61} proved that any $k$-edge connected graph has $k/2$ edge-disjoint spanning trees.
Since $(V,F)$ is $k$-edge connected, it has $k/2$ edge-disjoint spanning trees, and equivalently, $\{v_e\}_{e\in F}$ contains $k/2$ disjoint bases. 

Therefore, by \autoref{thm:thinbases},
 there exists a basis (i.e., a spanning tree) $T\subseteq F$ such that
\begin{equation} 
\norm{\sum_{e\in T} v_ev_e^\intercal } \leq \alpha,
\label{eq:LTnorm}
\end{equation}
for $\alpha=O(\eps+1/k)$. 
Fix an arbitrary vector $y\in\R^V$. We show that 
\begin{equation}
y^\intercal L_T y \leq \alpha\cdot y^\intercal L_G y,
\label{eq:LTnormy}
\end{equation}
and this completes the proof.
  By \eqref{eq:LTnorm} for any $x\in\R^V$,
  $$ x^\intercal \left(\sum_{e\in T} v_ev_e^\intercal\right) x \leq \alpha\cdot \norm{x}^2.$$
Let $x=L_G^{1/2}y$,  we get
$$ y^\intercal L_G^{1/2}  \left(L_G^{\dagger/2}\sum_{e\in T} b_eb_e^\intercal L_G^{\dagger/2}\right) L_G^{1/2} y \leq \alpha\cdot y^\intercal L_G y. $$
The above is the same as \eqref{eq:LTnormy} and we are done.
\end{proof}

The above corollary completely answers \autoref{pr:Feffres} but it is not enough for our purpose in \cite{AO14}; we need a slightly stronger statement. 
For a matrix $D\in\R^{V\times V}$ we say
$D\preceq_\square L_G,$ if for any set $S\subset V$,
$$ \bone_S^\intercal D \bone_S \leq \bone_S^\intercal L_G \bone_S,$$
where as usual $\bone_S\in\R^V$ is the indicator vector of the set $S$.
In the main theorem of \cite{AO14} we show that
for any $k$-edge-connected graph $G$ (for $k=7\log n$) there is a positive definite (PD) matrix $D\preceq_\square L_G$ and a set $F\subseteq E$ such that $(V,F)$ is $\Omega(k)$-edge-connected and 
$$ \max_{e\in F} b_e^\intercal D^{-1} b_e \leq \frac{\polylog(k)}{k}.$$
To show that $G$ has a combinatorially thin tree it is enough to show that there is a tree $T\subseteq E$ that is $\alpha$-spectrally thin w.r.t. $L_G+D$ for $\alpha=\polylog(k)/k$, i.e.,
$$ L_T \preceq \frac{\polylog(k)}{k} (L_G+D).$$
 Such a tree is $2\alpha$-combinatorially thin w.r.t. $G$ because $D\preceq_\square L_G$. 
Note that the above corollary does not imply $L_G+D$ has a spectrally thin tree because $D$ is not necessarily a Laplacian matrix.
Nonetheless, we can prove the existence of a spectrally thin tree with another application of \autoref{thm:thinbases}.
 
\begin{corollary}
\label{thm:spectralthintree}
Given a graph $G=(V,E)$, a PD matrix $D$, and $F\subseteq E$ such that $(V,F)$ is $k$-edge connected, if for any edge $e \in F$,
$$ b_e^\intercal D^{-1} b_e \leq \eps,$$
then
$G$ has a spanning tree $T\subseteq F$ such that
$$ L_T \preceq O(\eps+1/k)\cdot (L_G + D).$$
\end{corollary}
\begin{proof}
	The proof is very similar to \autoref{cor:mixedrealstability}.
	For any edge $e\in F$,
 let $v_e = (D+L_G)^{-1/2} b_e$. Note that since $D$ is PD, $D+L_G$ is PD and $(D+L_G)^{-1/2}$ is well defined.
 By the assumption, 
 $$\norm{v_e}^2 = b_e^\intercal (D+L_G)^{-1} b_e \leq b_e^\intercal D^{-1} b_e = \eps,$$
 where the inequality uses \autoref{lem:matrixinvineq}.
 In addition, the vectors are in sub-isotropic position,
 $$ \sum_{e\in F} v_e v_e^\intercal = (D+L_G)^{\dagger/2} L_F (D+L_G)^{\dagger/2} \preceq I.$$
The matrix PSD inequality uses that $L_F \preceq L_G \preceq D+L_G$. Furthermore, every basis of $\{v_e\}_{e\in E}$ is a spanning tree of $G$ and by $\Omega(k)$-connectivity of $F$, there are $\Omega(k)$-edge disjoint bases. 
 Therefore, by \autoref{thm:thinbases}, there is a tree $T\subseteq F$ such that
 $$ \norm{\sum_{e\in T} v_ev_e^\intercal} \leq \alpha,$$
 for $\alpha=O(\eps+1/k)$. Similar to \autoref{cor:mixedrealstability} this tree satisfies
 $$ L_T \preceq \alpha \cdot (L_G+D),$$
 and this completes the proof.
\end{proof}

\subsection{Proof Overview}
We build on the method of interlacing polynomials of 
\cite{MSS12,MSS13}.
Recall that an interlacing family of polynomials has the property that there is always a polynomial whose largest root is at most the largest root of the sum of the polynomials in the family.
First, we 
show that for any set of vectors assigned to the elements of a homogeneous strongly Rayleigh measure, the characteristic polynomials of natural quadratic forms associated with the samples of the distribution form an interlacing family.
This implies that there is a sample of the distribution such that the largest root of its characteristic polynomial is at most the largest root of the average of the characteristic polynomials of all samples of $\mu$. Then, we use  the multivariate barrier argument of \cite{MSS13} to upper-bound the largest root of our expected characteristic polynomial. 

Our proof has two main ingredients.
The first one is the construction of a new class of expected characteristic polynomials which are the  weighted average of the characteristic polynomials of the natural quadratic forms associated to the samples of the strongly Rayleigh distribution, where the weight of each polynomial is proportional to the probability of the corresponding sample set in the distribution. To show that the expected characteristic polynomial is real rooted we appeal to the theory of real stability. We show that our expected characteristic polynomial can be realized by applying $\prod_{i=1}^m(1-\partial/\partial^2_{z_i})$ operator to the real stable polynomial $g_\mu(z) \cdot \det(\sum_{i=1}^m z_i v_iv_i^\intercal)$, and then projecting all variables onto $x$. 

Our second ingredient is the extension of the multivariate barrier argument.
Unlike  \cite{MSS13}, here we need to prove an upper bound on the largest root of the mixed characteristic polynomial which is very close to zero. 
It turns out that the original idea of \cite{BSS14} 
that studies the behavior of the roots of a (univariate) polynomial $p(x)$ under the operator $1-\partial/\partial_x$ cannot establish upper bounds that are less than one.
Fortunately, here we need to study the behavior of the roots of a (multivariate) polynomial $p(z)$ under the operators $1-\partial/\partial^2_{z_i}$.
The $1-\partial/\partial^2_{z_i}$ operators allow us to impose very small shifts on the multivariate upper barrier assuming the  barrier functions are sufficiently small.
The intuition is that, since
$$ 1-\frac{\partial}{\partial^2_{z_i}} = \left(1-\frac{\partial}{\partial_{z_i}}\right)\cdot \left(1+\frac{\partial}{\partial_{z_i}}\right),$$
we  expect  $(1-\partial/\partial_{z_i})$ to shift the upper barrier by $1+\Theta(\delta)$ (for some $\delta$ depending on the value of the $i$-th barrier function) as proved in \cite{MSS13} and $(1+\partial/\partial_{z_i})$ to shift the upper barrier by $1-\Theta(\delta)$. Therefore, applying both operators  the upper barrier must be moved by no more than $\Theta(\delta)$.	

\section{Preliminaries}
We adopt a notation similar to \cite{MSS13}. We write ${[m]\choose k}$ to denote the collection of subsets of $[m]$ with exactly $k$ elements. 
We write $2^{[m]}$ to denote the family of all subsets of the set $[m]$. 
We write
$\partial_{z_i}$ to denote the operator that performs partial differentiation with respect to $z_i$. 
We use $\norm{v}$ to denote the Euclidean $2$-norm of a vector $x$. For a matrix $M$, we write $\norm{M}=\max_{\norm{x}=1}\norm{Mx}$ to denote the operator norm of $M$.
We use $\bone$ to denote the all $1$ vector. 
\subsection{Interlacing Families}
We  recall the definition of interlacing families of polynomials from \cite{MSS12}, and its main consequence.
\begin{definition}
We say that a real rooted polynomial $g(x) = \alpha_0 \prod_{i=1}^{m-1} (x-\alpha_i)$ interlaces a real rooted polynomial $f(x) = \beta_0 \prod_{i=1}^m (x-\beta_i)$ if
$$\beta_1 \leq \alpha_1 \leq \beta_2 \leq \alpha_2 \leq \ldots \leq \alpha_{m-1} \leq \beta_m. $$
\end{definition}
We say that polynomials $f_1,\ldots,f_k$ have a common interlacing if there is a polynomial $g$ such that $g$ interlaces all $f_i$.
The following lemma is proved in \cite{MSS12}.
\begin{lemma}
\label{lem:interlacringroot}
Let $f_1,\ldots,f_k$ be polynomials of the same degree that are real rooted and have positive leading coefficients. Define
$$ f_{\emptyset} = \sum_{i=1}^k f_i.$$
If $f_1,\ldots,f_k$ have a common interlacing, then there is an $i$ such that the largest root of $f_i$ is at most the largest root of $f_\emptyset$. 
\end{lemma}

\begin{definition}
\label{def:interlacingfamily}
Let $\cF\subseteq 2^{[m]}$ be nonempty. For any $S\in\cF$,  let $f_S(x)$ be a real rooted polynomial of degree $d$ with a positive leading coefficient.
For $s_1,\ldots,s_k\in\{0,1\}$ with $k<m$, let
$$ \cF_{s_1,\ldots,s_k}:=\{S\in\cF: i\in S \Leftrightarrow s_i=1\}.$$
Note that $\cF=\cF_{\emptyset}$.
Define
$$f_{s_1,\ldots,s_k} = \sum_{S\in\cF_{s_1,\ldots,s_k}} f_S,$$
and
$$ f_{\emptyset} = \sum_{S\in\cF} f_S.$$
We say polynomials $\{f_S\}_{S\in\cF}$ form an {\em interlacing family} if for all $0\leq k<m$ and all
$s_1,\ldots,s_k\in\{0,1\}$ the following holds: If both of  $\cF_{s_1,\ldots,s_k,0}$ and $\cF_{s_1,\ldots,s_k,1}$ are nonempty, $f_{s_1,\ldots,s_k,0}$ and $f_{s_1,\ldots,s_k,1}$ have 
a common interlacing. 
\end{definition}

The following is analogous to \cite[Thm 3.4]{MSS13}.
\begin{theorem}
\label{thm:interlacingfamily}
Let $\cF\subseteq 2^{[m]}$ and let $\{f_S\}_{S\in\cF}$ be an interlacing family of polynomials. Then, there exists $S\in\cF$ such that the largest root of $f(S)$ is at most the largest root of $f_\emptyset$.
\end{theorem}
\begin{proof}
We prove by induction. Assume that for some choice of $s_1,\ldots,s_k\in\{0,1\}$ (possibly with $k=0$), $\cF_{s_1,\ldots,s_k}$ is nonempty and the largest root of $f_{s_1,\ldots,s_k}$ is at most the largest root of $f_\emptyset$.
If  $\cF_{s_1,\ldots,s_k,0}=\emptyset$, then
$f_{s_1,\ldots,s_k}=f_{s_1,\ldots,s_k,1}$, so we let $s_{k+1}=1$ and we are done.
Similarly, if $\cF_{s_1,\ldots,s_k,1}=\emptyset$, then we let $s_{k+1}=0$ and we are done with the induction. If both of these sets are nonempty, then
 $f_{s_1,\ldots,s_k,0}$ and $f_{s_1,\ldots,s_k,1}$ have a common interlacing. So, by \autoref{lem:interlacringroot}, for some choice of $s_{k+1}\in\{0,1\}$, the largest root of
$f_{s_1,\ldots,s_{k+1}}$ is at most the largest root of $f_\emptyset$.
\end{proof}

We use the following lemma which appeared as Theorem 2.1 of \cite{Ded92} to prove that a certain family of polynomials that we construct in \autoref{sec:mixedpoly} form an interlacing family. 
\begin{lemma}
\label{lem:realrootedimplyinterlacing}
Let $f_1,\ldots,f_k$ be univariate polynomials of
the same degree with positive leading coefficients. Then, $f_1,\ldots,f_k$ have a common interlacing if and only if $\sum_{i=1}^k \lambda_i f_i$ is real rooted for all convex combinations $\lambda_i\geq 0$, $\sum_{i=1}^k \lambda_i=1$. 
\end{lemma}

\subsection{Stable Polynomials}
\label{sec:realstability}
Stable polynomials are natural multivariate generalizations
of real-rooted univariate polynomials. For a complex number $z$, let
$\image(z)$ denote the imaginary part of $z$.
We say a polynomial $p(z_1,\ldots,z_m)\in\C[z_1,\ldots,z_m]$ is {\em stable}
if whenever $\image(z_i)>0$ for all $1\leq i\leq m$, $p(z_1,\ldots,z_m)\neq 0$. We say $p(.)$ is real stable, if it is stable and all of its coefficients are real. It is easy to see that any univariate polynomial is real stable  if and only if it is real rooted.

One of the most interesting classes of real stable polynomials is the class of determinant polynomials as observed by Borcea and Br\"and\'en \cite{BB08}.
\begin{theorem}
\label{thm:motherrealstability}
For any set of positive semidefinite matrices $A_1,\ldots,A_m$,
the following polynomial is real stable:
$$ \det\Big(\sum_{i=1}^m z_i A_i\Big).$$
\end{theorem}

Perhaps the most important property of stable polynomials is that they are closed under several elementary operations like multiplication, differentiation, and substitution. We will use these operations to generate new stable polynomials from the determinant polynomial.
The following is proved in \cite{MSS13}.
\begin{lemma}
\label{lem:partialrealstability}
If $p\in\R[z_1,\ldots,z_m]$ is real stable, then so are the polynomials $(1-\partial_{z_1})p(z_1,\ldots,z_m)$ and
$(1+\partial_{z_1})p(z_1,\ldots,z_m)$.
\end{lemma}
The following corollary simply follows from the above lemma.
\begin{corollary}
\label{cor:partial2realstability}
If $p\in\R[z_1,\ldots,z_m]$ is real stable, then so is 
$$(1-\partial^2_{z_1})p(z_1,\ldots,z_m).$$
\end{corollary}
\begin{proof}
First, observe that
$$(1-\partial^2_{z_1})p(z_1,\ldots,z_m) = (1-\partial_{z_1})(1+\partial_{z_1})p(z_1,\ldots,z_m). $$
So, the conclusion follows from two applications of \autoref{lem:partialrealstability}.
\end{proof}

The following closure properties are elementary.
\begin{lemma}
\label{lem:stabilitymultiplication}
If $p\in\R[z_1,\ldots,z_m]$ is real stable, then so is $p(\lambda\cdot z_1,\ldots,\lambda_m\cdot z_m)$ for  real-valued $\lambda_1,\ldots,\lambda_m > 0$. 
\end{lemma}
\begin{proof}
Say $(z_1,\ldots,z_m)\in\C^m$ is a root of $p(\lambda\cdot z_1,\ldots,\lambda_m\cdot z_m)$. Then $(\lambda_1\cdot z_1,\ldots,\lambda_m\cdot z_m)$ is a root of $p(z_1,\ldots,z_m)$.
Since $p$ is real stable, there is an $i$ such that $\image(\lambda_i\cdot z_i)\leq 0$. But, since $\lambda_i>0$, we get $\image(z_i)\leq 0$, as desired.
\end{proof}
\begin{lemma}
\label{lem:stabilityaddition}
If $p\in\R[z_1,\ldots,z_m]$ is real stable, then so is
$p(z_1+x,\ldots,z_m+x)$ for a new variable $x$.
\end{lemma}
\begin{proof}
Say $(z_1,\ldots,z_m,x)\in\C^m$ is a root of $p(z_1+x,\ldots,z_m+x)$. Then $(z_1+x,\ldots,z_m+x)$ is a root of $p(z_1,\ldots,z_m)$. 
Since $p$ is real stable, there is an $i$ such that $\image(z_i+x)\leq 0$. But, then either $\image(x)\leq 0$ or $\image(z_i)\leq 0$, as desired. 
\end{proof}

\subsection{Facts from Linear Algebra}
For a Hermitian matrix $M\in\C^{d\times d}$,
we write the characteristic polynomial of $M$ in terms of a variable $x$ as
$$ \X[M](x)=\det(xI-M).$$
We also write the characteristic polynomial in terms of the square of $x$ as $$\X[M](x^2)=\det(x^2I-M).$$

For $1\leq k\leq n$, we write $\sigma_k(M)$
to denote the sum of all principal $k\times k$ minors of $M$, in particular,
$$ \X[M](x) = \sum_{k=0}^d x^{d-k} (-1)^k\sigma_k(M).$$

The following lemma follows from the Cauchy-Binet identity. See \cite{MSS13} for the proof.
\begin{lemma}
\label{lem:cauchbinet}
For vectors $v_1,\ldots,v_m\in\R^d$ and scalars $z_1,\ldots,z_m$,
$$ \det\left(xI + \sum_{i=1}^m z_iv_iv_i^\intercal\right) = \sum_{k=0}^d x^{d-k}  \sum_{S\subseteq {[m]\choose k}} z^S \sigma_k\big(\sum_{i\in S} v_iv_i^\intercal\big).$$
In particular, for $z_1=\ldots=z_m=-1$,
$$ \det\left(xI - \sum_{i=1}^m v_iv_i^\intercal\right) = \sum_{k=0}^d x^{d-k} (-1)^k \sum_{S\subseteq {[m]\choose k}} \sigma_k\big(\sum_{i\in S} v_iv_i^\intercal\big).$$
\end{lemma}

The following is Jacboi's formula for the derivative of the determinant of a matrix.
\begin{theorem}
\label{thm:jacobi}
For an invertible matrix $A$ which is a differentiable function of $t$,
$$ \partial_t \det(A) = \det(A)\cdot \trace(A^{-1} \partial_t A).$$
\end{theorem}
\begin{lemma}
\label{lem:inverse}
For an invertible matrix $A$ which is a differentiable function of $t$,
$$ \frac{\partial A^{-1}}{\partial t} = -A^{-1} (\partial_t A) A^{-1}.$$
\end{lemma}
\begin{proof}
Differentiating both sides of the identity $A^{-1} A = I$ with respect to $t$, we get
$$ A^{-1} \frac{\partial A}{\partial t} + \frac{\partial A^{-1}}{\partial t} A = 0. $$
Rearranging the terms and multiplying with $A^{-1}$ gives the lemma's conclusion.
\end{proof}

The following two standard facts about trace will be used throughout the paper.
First, for $A\in\R^{k\times n}$ and $B\in\R^{n\times k}$,
$$ \trace(AB)=\trace(BA).$$
Secondly, for positive semidefinite matrices $A,B$ of the same dimension,
$$ \trace(AB)\geq 0.$$
Also, we use the fact that for any positive semidefinite matrix $A$ and any Hermitian matrix $B$,
$BAB$ is positive semidefinite. 
\begin{lemma}
\label{lem:matrixinvineq}
If $A,B\in\R^{n\times n}$ are PD matrices and  $A\preceq B$, then $B^{-1} \preceq A^{-1}.$
\end{lemma}
\begin{proof}
Since $A\preceq B$,
$$ B^{-1/2} A B^{-1/2} \preceq B^{-1/2} B B^{-1/2} = I.$$
So, 
$$ B^{1/2} A^{-1} B^{1/2} = (B^{-1/2} A B^{-1/2})^{-1} \succeq I$$
Multiplying both sides of the above by $B^{-1/2}$, we get
$$ A^{-1} = B^{-1/2} B^{1/2} A^{-1} B^{1/2} B^{-1/2} \succeq B^{-1/2} I B^{-1/2} = B^{-1}.$$
\end{proof}

\section{The Mixed Characteristic Polynomial}
\label{sec:mixedpoly}
For a probability distribution $\mu$, let $d_\mu$ be the degree of the polynomial $g_\mu$.
\begin{theorem}
For $v_1,\ldots,v_m\in\R^d$ and a homogeneous probability distribution $\mu:[m]\to\R_+$,
\begin{equation}
\label{eq:mixedcharpoly}
x^{d_\mu-d}\uE{S\sim \mu}\X\left[\sum_{i\in S} 2v_iv_i^\intercal\right] (x^2) = \prod_{i=1}^m \left(1-\partial^2_{z_i}\right) \left(g_{\mu}(x\bone +z) \cdot \det\left(xI+\sum_{i=1}^m z_i v_iv_i^\intercal \right)\right) \Bigg|_{z_1=\ldots=z_m=0}.
\end{equation}
\end{theorem}
We call the polynomial $\iE{S\sim\mu}\X[\sum_{i\in S} 2v_iv_i^\intercal](x^2)$ the {\em mixed characteristic polynomial} and we denote it by $\mu[v_1,\ldots,v_m](x)$.
\begin{proof}
For $S\subseteq [m]$, let $z^{2S}=\prod_{i\in S} z_i^2.$
By \autoref{lem:cauchbinet}, the coefficient of $z^{2S}$ in 
$$g_\mu(x\bone +z)\cdot \det(xI+\sum_{i=1}^m z_iv_iv_i^\intercal)$$ 
is equal to
$$\left(\prod_{i\in S} \partial^2_{z_i}\right) \left(g_{\mu}(x\bone +z)\cdot \det\left(xI+\sum_{i=1}^m z_i v_iv_i^\intercal\right)\right)\Bigg|_{z_1=\ldots=z_m=0}.
$$
Each of the two polynomials $g_\mu(x\bone +z)$ and $\det(xI+\sum_{i=1}^m z_i v_iv_i^\intercal)$ is multi-linear in $z_1,\ldots,z_m$. Therefore, for $k=|S|$, the above is equal to 
\begin{equation}
\label{eq:firstderivmixedpoly}
 2^k\cdot \left(\prod_{i\in S} \partial_{z_i}\right) g_\mu(x\bone +z) \Bigg|_{z_1=\ldots=z_m=0} \cdot \left(\prod_{i\in S} \partial_{z_i}\right) \det\left(xI+\sum_{i=1}^m z_i v_iv_i^\intercal\right)\Bigg|_{z_1=\ldots=z_m=0}. 
 \end{equation}
Since  $g_\mu$ is a homogeneous  polynomial of degree $d_\mu$, the first term in the above is equal to
$$ 
x^{d_\mu-k} \PP{T\sim\mu}{S\subseteq T}.$$
And, by \autoref{lem:cauchbinet},
the second term of \eqref{eq:firstderivmixedpoly} is equal to
$$ 
x^{d-k} \sigma_k\left(\sum_{i\in S} z_iv_iv_i^\intercal\right). $$
Applying the above identities for all $S\subseteq [m]$,
\begin{align*}
\prod_{i=1}^m \left(1-\partial^2_{z_i}\right) \Bigg(g_{\mu}(x\bone +z) &\cdot \det\left(xI+\sum_{i=1}^m z_i v_iv_i^\intercal \right) \Bigg)\Bigg|_{z_1=\ldots=z_m=0}\\
&=\sum_{k=0}^m (-1)^k \sum_{S\subseteq {[m]\choose k}} \left(\prod_{i\in S} \partial^2_{z_i}\right) \left(g_\mu(x\bone +z)\cdot\det\left(xI+\sum_{i=1}^m z_i v_iv_i^\intercal\right)\right)\Bigg|_{z_1=\ldots=z_m=0}\\
&=\sum_{k=0}^d (-1)^k 2^k x^{d_\mu+d-2k}\sum_{S\in {[m]\choose k}} \PP{T\sim\mu}{S\subseteq T}\cdot\sigma_k\left(\sum_{i\in S} v_iv_i^\intercal\right)\\
&=x^{d_\mu-d} \uE{S\sim\mu}\X\left[\sum_{i\in S}2v_iv_i^\intercal\right](x^2).
\end{align*}
The last identity uses \autoref{lem:cauchbinet}.
\end{proof}

\begin{corollary}
\label{cor:mixedrealstability}
If $\mu$ is a strongly Rayleigh probability distribution, then the mixed characteristic polynomial is real-rooted.
\end{corollary}
\begin{proof}
First, by \autoref{thm:motherrealstability},
$$ \det\left(xI+\sum_{i=1}^m z_iv_iv_i^\intercal\right)$$
is real stable. Since $\mu$ is strongly Rayleigh, $g_\mu(z)$ is real stable. So, by \autoref{lem:stabilityaddition}, $g_\mu(x\bone+z)$ is real stable. The product of two real stable polynomials is also real stable, so
$$ g_\mu(x\bone +z)\cdot\det\left(xI+\sum_{i=1}^m z_iv_iv_i^\intercal\right)$$
is real stable.
\autoref{cor:partial2realstability} implies that
$$ \prod_{i=1}^m \left(1-\partial^2_{z_i}\right) \left(g_\mu(x\bone +z)\cdot\det\left(xI+\sum_{i=1}^mz_iv_iv_i^\intercal\right)\right)$$
is real stable as well.
Wagner \cite[Lemma 2.4(d)]{Wag11}
tells us that real stability is preserved under setting variables to real numbers, so
$$ \prod_{i=1}^m \left(1-\partial^2_{z_i}\right) \left(g_\mu(x\bone +z)\cdot\det\left(xI+\sum_{i=1}^mz_iv_iv_i^\intercal\right)\right)\Bigg|_{z_1=\ldots=z_m=0}$$
is a univariate real-rooted polynomial.
The mixed characteristic polynomial is equal to the above polynomial up to a term $x^{d_\mu-d}$.
So, the mixed characteristic polynomial is also real  rooted.
\end{proof}

Now, we use the real-rootedness of the mixed characteristic polynomial to show that the characteristic polynomials of the set of vectors assigned to any set $S$ with nonzero probability in $\mu$ form an interlacing family. 
For a homogeneous strongly Rayleigh measure $\mu$, let  
$$\cF=\{S: \mu(S) >0\},$$
and for $s_1,\ldots,s_k\in\{0,1\}$ let $\cF_{s_1,\ldots,s_k}$ be as defined in \autoref{def:interlacingfamily}.
For any $S\in\cF$, let
$$q_S(x) = \mu(S)\cdot \X\left[\sum_{i\in S} 2v_iv_i^\intercal \right](x^2).$$

\begin{theorem}
\label{thm:mixedinterlacing}
The polynomials $\{q_S\}_{S\in\cF}$ form an interlacing family.
\end{theorem}
\begin{proof}
For $1\leq k\leq m$ and $s_1,\ldots,s_k\in\{0,1\}$,
let $\mu_{s_1,\ldots,s_k}$ be $\mu$ conditioned on
the sets $S\in\cF_{s_1,\ldots,s_k}$, i.e., $\mu$ conditioned on $i\in S$ for all $i\leq k$ where $s_i=1$ and $i\notin S$ for all $i\leq k$ where $s_i=0$. 
We inductively write the generating polynomial of $\mu_{s_1,\ldots,s_k}$ in terms of $g_\mu$.
Say we have written $g_{\mu_{s_1,\ldots,s_k}}$ in terms of $g_\mu$. Then, we can write,
\begin{eqnarray}
\label{eq:musk1}
g_{\mu_{s_1,\ldots,s_k,1}}(z) &=&
\frac{z_{k+1}\cdot\partial_{z_{k+1}} g_{\mu_{s_1,\ldots,s_k}}(z)}{\partial_{z_{k+1}}
g_{\mu_{s_1,\ldots,s_k}}(z)\big|_{z_i=1}},\\
g_{\mu_{s_1,\ldots,s_k,0}}(z) &=& \frac{g_{\mu_{s_1,\ldots,s_k}}(z) \big|_{z_{k+1}=0}}{g_{\mu_{s_1,\ldots,s_k}}(z)\big|_{z_{k+1}=0, z_i=1 \text{ for } i\neq k+1}}.
\label{eq:musk0}
\end{eqnarray}
Note that the denominators of both equations are just normalizing constants. 
The above polynomials are well defined if the normalizing constants are nonzero, i.e., if the set $\cF_{s_1,\ldots,s_k,s_{k+1}}$ is nonempty.
Since the real stable polynomials are closed under differentiation and substitution, for any $1\leq k\leq m$, and $s_1,\ldots,s_k\in\{0,1\}$, if $g_{\mu_{s_1,\ldots,s_k}}$ is well defined,
it is real stable, so $\mu_{s_1,\ldots,s_k}$ is a strongly Rayleigh distribution.

Now, for $s_1,\ldots,s_k\in\{0,1\}$, let
$$ q_{s_1,\ldots,s_k}(x) = \sum_{S\in\cF_{s_1,\ldots,s_k}} q_S(x). $$
Since $\mu_{s_1,\ldots,s_k}$ is strongly Rayleigh,  by \autoref{cor:mixedrealstability},
$q_{s_1,\ldots,s_k}(x)$ is real rooted.

By \autoref{lem:realrootedimplyinterlacing}, to prove the theorem it is enough to show that if $\cF_{s_1,\ldots,s_k,0}$ and $\cF_{s_1,\ldots,s_k,1}$ are nonempty, then for any $0 < \lambda < 1$,
$$ \lambda\cdot q_{s_1,\ldots,s_k,1}(x) + (1-\lambda) \cdot q_{s_1,\ldots,s_k,0}(x)$$
is real rooted.
Equivalently, by \autoref{cor:mixedrealstability},
it is enough to show that for any $0<\lambda<1$,
\begin{equation} \lambda\cdot g_{\mu_{s_1,\ldots,s_k,1}}(z) + (1-\lambda) \cdot g_{\mu_{s_1,\ldots,s_k,0}}(z) 
\label{eq:lambdagrealstable}
\end{equation}
is real stable.
Let us write,
\begin{eqnarray*}
g_{\mu_{s_1,\ldots,s_k}}(z) &=&  z_{k+1}\cdot\partial_{z_{k+1}} g_{\mu_{s_1,\ldots,s_k}}(z) + g_{\mu_{s_1,\ldots,s_k}}(z)\big|_{z_{k+1}=0} \\
&=& \alpha\cdot g_{\mu_{s_1,\ldots,s_k,1}}(z) + \beta\cdot g_{\mu_{s_1,\ldots,s_k,0}}(z),
\end{eqnarray*}
for some $\alpha,\beta>0$. The second identity follows by \eqref{eq:musk1} and \eqref{eq:musk0}.
Let $\lambda_{k+1} > 0$ such that
\begin{equation} \frac{\lambda_{k+1}\cdot \alpha}{\lambda} = \frac{\beta}{1-\lambda}.
\label{eq:lambdakp1def}
\end{equation}
Since $g_{\mu_{s_1,\ldots,s_k}}$ is real stable,
by \autoref{lem:stabilitymultiplication}
$$ g_{\mu_{s_1,\ldots,s_k}}(z_1,\ldots,z_k,\lambda_{k+1}\cdot z_{k+1},z_{k+2},\ldots,z_m)$$
is real stable.
But, by \eqref{eq:lambdakp1def} the above polynomial is just a multiple of  \eqref{eq:lambdagrealstable}. So, \eqref{eq:lambdagrealstable} is real stable.
\end{proof}

\section{An Extension of \cite{MSS13} Multivariate Barrier Argument}
\label{sec:barrier}
In this section we upper-bound the roots of the mixed characteristic polynomial in terms of the marginal probabilities of elements of $[m]$ in $\mu$ and the maximum of the squared norm of vectors $v_1,\ldots,v_m$.
\begin{theorem}
\label{thm:barriermaxroot}
Given vectors $v_1,\ldots,v_m\in\R^d$, and a homogeneous strongly Rayleigh probability distribution $\mu:[m]\to\R_+$, such that the marginal probability of each element $i\in [m]$ is at most $\eps_1$, $\sum_{i=1}^m v_iv_i^\intercal=I$ and $\norm{v_i}^2 \leq \eps_2$,
the largest root of $\mu[v_1,\ldots,v_m](x)$
is at most $4(2\eps+\eps^2)$, where $\eps=\eps_1+\eps_2$,
\end{theorem}

First, similar to \cite{MSS13} we derive a slightly different expression.
\begin{lemma}
For any probability distribution $\mu$ and vectors $v_1,\ldots,v_m\in\R^d$ such that $\sum_{i=1}^m v_iv_i^\intercal=I$, 
$$ x^{d_\mu-d}\mu[v_1,\ldots,v_m](x)=\prod_{i=1}^m \left(1-\partial^2_{y_i}\right)\left(g_\mu(y)\cdot\det\left(\sum_{i=1}^m y_iv_iv_i^\intercal\right)\right)\Bigg|_{y_1=\ldots=y_m=x}. $$
\end{lemma}
\begin{proof}
This is because for any differentiable function
$f$, 
$\partial_{y_i} f(y_i)|_{y_i=z_i+x} = \partial_{z_i} f(z_i+x).$
\end{proof}

Let
$$ Q(y_1,\ldots,y_m)=\prod_{i=1}^m \left(1-\partial^2_{y_i}\right)\left(g_\mu(y) \cdot\det\left(\sum_{i=1}^my_iv_iv_i^\intercal\right)\right).$$
Then, by the above lemma, the maximum root of $Q(x,\ldots,x)$ is the same as the maximum root of $\mu[v_1,\ldots,v_m](x)$. In the rest of this section we upper-bound the maximum root of $Q(x,\ldots,x)$.

It directly follows from the proof of  Theorem 5.1 in \cite{MSS13} that the maximum root of $Q(x,\ldots,x)$  is at most $(1+\sqrt{\eps})^2$. 
But, in our setting, any upper-bound that is more than 1 obviously holds, as for any $S\subseteq[m]$,
$$ \norm{\sum_{i=1}^m v_iv_i^\intercal} \leq 1.$$
The main difficulty that we are facing is to prove an upper-bound  of $O(\eps)$ on the maximum root of $Q(x,\ldots,x)$. 

We use an extension of the multivariate barrier argument of \cite{MSS13} to upper-bound the maximum root of $Q$. We manage to prove a significantly smaller upper-bound because we apply $1-\partial^2_{y_i}$ operators as opposed to the $1-\partial_{y_i}$ operators used in \cite{MSS13}. This allows us to impose significantly smaller shifts on the barrier upper-bound in our inductive argument.

\begin{definition}
For a multivariate polynomial $p(z_1,\ldots,z_m)$,
we say $z\in\R^m$ is above all roots of $p$
if for all $t\in\R^m_+$,
$$ p(z+t)> 0. $$
We use $\Ab_p$ to denote the set of points which
are above all roots of $p$.
\end{definition}

We use the same barrier function defined in \cite{MSS13}.
\begin{definition}
For a real stable polynomial $p$, and $z\in\Ab_p$, the barrier function of $p$ in direction $i$ at $z$ is
$$ \Phi^i_p(z) := \frac{\partial_{z_i}p(z)}{p(z)}= \partial_{z_i}\log p(z).$$
To analyze the rate of change of the barrier function with respect to the $1-\partial^2_{z_i}$ operator, we need to work with the second derivative of $p$ as well. We define,
$$ \Psi^i_p(z) := \frac{\partial^2_{z_i}p(z)}{p(z)}.$$
Equivalently, for a univariate restriction $q_{z,i}(t) = p(z_1,\ldots,z_{i-1},t,z_{i+1},\ldots,z_m)$,
with real roots $\lambda_1,\ldots,\lambda_r$ we can write,
\begin{align*}
\Phi^i_p(z) &= \frac{q_{z,i}'(z_i)}{q_{z,i}(z_i)}=
\sum_{j=1}^r \frac{1}{z_i-\lambda_j},\\
\Psi^i_p(z)&= \frac{q_{z,i}''(z_i)}{q_{z,i}(z_i)} = \sum_{1\leq j<k\leq r} \frac{2}{(z_i-\lambda_j)(z_i-\lambda_k)}. 
\end{align*}
\end{definition}

The following lemma is immediate from the above definition.
\begin{lemma}
\label{lem:psiphi}
If $p$ is real stable and $z\in\Ab_p$, then for all $i\leq m$,
$$ \Psi^i_p(z) \leq  \Phi^i_p(z)^2.$$
\end{lemma}
\begin{proof}
Since $z\in\Ab_p$, $z_i > \lambda_j$ for all $1\leq j\leq r$, so,
$$ \Phi^i_p(z)^2 - \Psi^i_p(z) = \left(\sum_{j=1}^r \frac{1}{z_i-\lambda_j}\right)^2 - \sum_{1\leq j<k\leq r}\frac{2}{(z_i-\lambda_j)(z_i-\lambda_k)} = \sum_{j=1}^r \frac{1}{(z_i-\lambda_j)^2} >0.$$
\end{proof}

The following monotonicity and convexity properties of the barrier functions are proved in \cite{MSS13}.
\begin{lemma}
\label{lem:barrierprops}
Suppose $p(.)$ is a real stable polynomial and $z\in\Ab_p$. Then, for all $i,j\leq m$ and $\delta\geq 0$,
\begin{align}
\label{eq:monotonicity}
&\Phi^i_p(z+\delta \bone_j) \leq \Phi^i_p(z) \text{ and},~~~~~~~~~~~~~~~~~~~~~\text{(monotonicity)}\\
&\Phi^i_p(z+\delta\bone_j) \leq \Phi^i_p(z) + \delta\cdot\partial_{z_j}\Phi^i_p(z+\delta\bone_j)~~~~~~~\text{(convexity)}.
\label{eq:convexity}
\end{align}
\end{lemma}

Recall that the purpose of the barrier functions $\Phi^i_p$ is to allow us to reason about the relationship between $\Ab_p$ and $\Ab_{p-\partial^2_{z_i}p}$; the monotonicity property and 
\autoref{lem:psiphi} imply the following lemma.
\begin{lemma}
\label{lem:rootsandbarrier}
If $p$ is real stable and $z\in\Ab_p$ is such that $\Phi^i_p(z)<1$, then $z\in\Ab_{p-\partial^2_{z_i} p}$.
\end{lemma}
\begin{proof}
Fix a nonnegative vector $t$. Since $\Phi$ is nonincreasing in each coordinate, 
$$\Phi^i_p(z+t)\leq \Phi^i_p(z) <1.$$
Since $z+t\in\Ab_p$, by \autoref{lem:psiphi},
$$ \Psi^i_p(z+t)\leq \Phi^i_p(z+t)^2 < 1. $$
Therefore, 
$$\partial^2_{z_i} p(z+t) < p(z+t) \Rightarrow (1-\partial^2_{z_i})p(z+t) >0, $$
as desired.
\end{proof}

We use an inductive argument similar to \cite{MSS13}.
We argue that when we apply each operator $(1-\partial^2_{z_j})$, the barrier functions, $\Phi^i_p(z)$, do not increase by shifting the upper bound along the direction $\bone_j$. 
As we would like to prove a significantly smaller upper bound on the maximum root of the mixed characteristic polynomial, we may only  shift along direction $\bone_j$ by a small amount. 
In the following lemma we show that
 when we apply the $(1-\partial^2_{z_j})$ operator we only need to shift the upper bound proportional to $\Phi^j_p(z)$ along the direction $\bone_j$.
\begin{lemma}
\label{lem:barriershift}
Suppose that $p(z_1,\ldots,z_m)$ is real stable and $z\in\Ab_p$. If for $\delta>0$,
$$ \frac2{\delta}\Phi^j_p(z)+\Phi^j_p(z)^2 \leq 1,$$
then, for all $i$,
$$ \Phi^i_{p-\partial^2_{z_j}p}(z+\delta\cdot\bone_j) \leq \Phi^i_p(z).$$
\end{lemma}

\noindent To prove the above lemma we first need to prove a technical lemma to upper-bound 
$\frac{\partial_{z_i}\Psi^j_p(z)}{\partial_{z_i}\Phi^j_p(z)}$.
We use the following characterization of the bivariate real stable polynomials proved by Lewis, Parrilo, and Ramana \cite{LPR05}.
The following form is stated in \cite[Cor 6.7]{BB10}.

\begin{lemma}
\label{lem:lax}
If $p(z_1,z_2)$ is a bivariate real stable polynomial of degree $d$, then there exist $d\times d$ positive semidefinite matrices $A,B$ and a Hermitian matrix $C$ such that
$$ p(z_1,z_2) = \pm \det(z_1A+z_2B+C).$$
\end{lemma}

\begin{lemma}
\label{lem:barrierhigherderivative}
Suppose that $p$ is real stable and $z\in\Ab_p$,
then for all $i,j\leq m$,
$$ \frac{\partial_{z_i}\Psi^j_p(z)}{\partial_{z_i}\Phi^j_p(z)} \leq 2\Phi^j_p(z).$$
\end{lemma}
\begin{proof}
If $i=j$, then we consider the univariate restriction $q_{z,i}(z_i)  = \prod_{k=1}^r (z_i-\lambda_k)$.
Then, 
\begin{eqnarray*}
\frac{\partial_{z_i}\sum_{1\leq k<\l\leq r} \frac{2}{(z_i-\lambda_k)(z_i-\lambda_\l)}}{\partial_{z_i}\sum_{k=1}^r \frac{1}{(z_i-\lambda_k)}} = \frac{\sum_{k\neq \l} \frac{-2}{(z_i-\lambda_k)^2(z_i-\lambda_\l)}}{\sum_{k=1}^r \frac{-1}{(z_i-\lambda_k)^2}} \leq \sum_{\l=1}^r \frac{2}{(z_i-\lambda_\l)} = 2\Phi^j_p(z). 
\end{eqnarray*}
The inequality uses the assumption that $z\in\Ab_p$.

If $i\neq j$, we fix all variables other than $z_i,z_j$ and we consider the bivariate restriction
$$ q_{z,ij}(z_i,z_j)=p(z_1,\ldots,z_m).$$
By \autoref{lem:lax}, 
there are Hermitian positive semidefinite matrices $B_i,B_j$, and a Hermitian matrix $C$ such that
$$q_{z,ij}(z_i,z_j) = \pm \det(z_iB_i+z_jB_j+C). $$
Let $M=z_iB_i +z_jB_j+C$. Marcus, Spielman, and Srivastava \cite[Lem 5.7]{MSS13} observed that
the sign is always positive, that $B_i+B_j$ is 
positive definite. In addition,  $M$ is  positive definite since $B_i+B_j$ is positive definite and $z\in\Ab_p$.

By \autoref{thm:jacobi}, the barrier function in direction $j$ can be expressed as
\begin{equation}
\label{eq:Philax}
\Phi^j_p(z) =  \frac{\partial_{z_j}\det(M)}{\det(M)} = \frac{\det(M)\trace(M^{-1}B_j)}{\det(M)}=\trace(M^{-1}B_j).
\end{equation}
By another application of \autoref{thm:jacobi},
\begin{eqnarray*} 
\Psi^j_p(z)=\frac{\partial^2_{z_j}\det(M)}{\det(M)} &=& \frac{\partial_{z_j}(\det(M)\trace(M^{-1}B_j))}{\det(M)} \\
&=& \frac{\det(M)\trace(M^{-1}B_j)^2}{\det(M)} + \frac{\det(M)\trace((\partial_{z_j}M^{-1})B_j)}{\det(M)}\\
&=& \trace(M^{-1}B_j)^2+\trace(-M^{-1}B_jM^{-1}B_j)\\
&=&\trace(M^{-1}B_j)^2 -\trace((M^{-1}B_j)^2).
\end{eqnarray*}
The second to last identity uses \autoref{lem:inverse}.
Next, we calculate $\partial_{z_i}\Phi^j_p$ and $\partial_{z_i}\Psi^j_p$. First,
by another application of \autoref{lem:inverse},
$$ \partial_{z_i} M^{-1}B_j = -M^{-1} B_i M^{-1} B_j=:L.$$
Therefore,
$$ \partial_{z_i} \Phi^j_p(z) = \partial_{z_i}\trace(M^{-1}B_j) = \trace(L),$$
and
\begin{eqnarray*} 
\partial_{z_i} \Psi^j_p(z) &=& \partial_{z_i}\trace(M^{-1}B_j)^2 - \partial_{z_i}\trace((M^{-1}B_j)^2) \\
&=& 2\trace(M^{-1}B_j)\trace(L) - \trace\left(L(M^{-1}B_j) + (M^{-1}B_j)L \right)\\
&=& 2\trace(M^{-1}B_j)\trace(L) - 2\trace(LM^{-1}B_j).
\end{eqnarray*}
Putting above equations together we get
\begin{eqnarray*}\frac{\partial_{z_i} \Psi^j_p(z)}{\partial_{z_i} \Phi^j_p(z)} &=& 2\frac{\trace(M^{-1}B_j)\trace(L) - \trace(LM^{-1}B_j)}{\trace(L)} \\
&=& 2\trace(M^{-1}B_j) - 2\frac{\trace(LM^{-1}B_j)}{\trace(L)}\\
&=& 2\Phi^j_p(z) - 2\frac{\trace(LM^{-1}B_j)}{\trace(L)} 
\end{eqnarray*}
where we used \eqref{eq:Philax}.

To prove the lemma it is enough to show that $\frac{\trace(LM^{-1}B_j)}{\trace(L)} \geq 0$. We show that both the numerator and the denominator are nonpositive.  First, 
$$ \trace(L)=-\trace(M^{-1}B_iM^{-1}B_j)\leq 0$$
where we used that $M^{-1}B_iM^{-1}$ and $B_j$ are positive semidefinite and the fact that the trace of the product of positive semidefinite matrices is nonnegative.
Secondly,
$$ \trace(LM^{-1}B_j) = \trace(-M^{-1}B_iM^{-1}B_jM^{-1}B_j) = -\trace(B_iM^{-1}B_jM^{-1}B_jM^{-1}) \leq 0,$$
where we again used that $M^{-1}B_jM^{-1}B_jM^{-1}$ and $B_i$ are positive semidefinite and the trace of the product of two positive semidefinite matrices is nonnegative.
\end{proof}

\begin{proofof}{\autoref{lem:barriershift}}
We write $\partial_i$ instead of $\partial_{z_i}$ for the ease of notation.
First, we write $\Phi^i_{p-\partial^2_j p}$
in terms of $\Phi^i_p$ and $\Psi^j_p$ and $\partial_i\Psi^j_p$.
\begin{eqnarray*}
\Phi^i_{p-\partial^2_j p} &=& \frac{\partial_i(p-\partial^2_j p)}{p-\partial^2_j p}\\
&=& \frac{\partial_i((1-\Psi^j_p)p)}{(1-\Psi^j_p)p}\\
&=& \frac{(1-\Psi^j_p)(\partial_i p)}{(1-\Psi^j_p)p} + \frac{(\partial_i(1-\Psi^j_p))p}{(1-\Psi^j_p)p}\\
&=& \Phi^i_p - \frac{\partial_i\Psi^j_p}{1-\Psi^j_p}.
\end{eqnarray*}
We would like to show that $\Phi^i_{p-\partial^2_j p}(z+\delta\bone_j)\leq \Phi^i_p(z)$. Equivalently, it is enough to show that
$$ -\frac{\partial_i \Psi^j_p(z+\delta\bone_j)}{1-\Psi^j_p(z+\delta\bone_j)} \leq \Phi^i_p(z) - \Phi^i_p(z+\delta\bone_j). $$
By \eqref{eq:convexity} of \autoref{lem:barrierprops},
it is enough to show that
$$ -\frac{\partial_i\Psi^j_p(z+\delta\bone_j)}{1-\Psi^j_p(z+\delta\bone_j)} \leq \delta\cdot(-\partial_j\Phi^i_p(z+\delta\bone_j)).$$
By \eqref{eq:monotonicity} of \autoref{lem:barrierprops}, $\delta\cdot(-\partial_j\Phi^i_p(z+\delta\bone_j)) >0$ so we may divide both sides of the above inequality by this term and obtain
$$ \frac{-\partial_i\Psi^j_p(z+\delta\bone_j)}{-\delta\cdot\partial_i\Phi^j_p(z+\delta\bone_j)} \cdot \frac{1}{1-\Psi^j_p(z+\delta\bone_j)} \leq 1,$$
where we also used $\partial_j\Phi^i_p=\partial_i\Phi^j_p$.
By \autoref{lem:barrierhigherderivative},
$ \frac{\partial_i \Psi^j_p}{\partial_i \Phi^j_p} \leq 2\Phi^j_p$. So, we can write,
$$ \frac2{\delta} \Phi^j_p(z+\delta\bone_j)\cdot \frac{1}{1-\Psi^j_p(z+\delta\bone_j)}\leq 1.$$
By \autoref{lem:psiphi} and \eqref{eq:monotonicity}
of \autoref{lem:barrierprops},
\begin{align*}
\Phi^j_p(z+\delta\bone_j) &\leq \Phi^j_p(z),\\
 \Psi^j_p(z+\delta\bone_j)&\leq \Phi^j_p(z+\delta\bone_j)^2\leq \Phi^j_p(z)^2.
\end{align*}
So, it is enough to show that
$$ \frac2{\delta}\Phi^j_p(z)\cdot\frac{1}{1-\Phi^j_p(z)^2} \leq 1$$
Using $\Phi^j_p(z) <1$ we may multiply both sides with $1-\Phi^j_p(z)$ and we obtain,
$$ \frac2{\delta}\Phi^j_p(z) + \Phi^j_p(z)^2 \leq 1,$$
as desired.
\end{proofof}

Now, we are read to prove \autoref{thm:barriermaxroot}.
\begin{proofof}{\autoref{thm:barriermaxroot}}
Let 
$$p(y_1,\ldots,y_m)=g_\mu(y) \cdot \det\left(\sum_{i=1}^m y_iv_iv_i^\intercal\right).$$
Set
$\eps=\eps_1+\eps_2$ and 
$$ \delta=t=\sqrt{2\eps+\eps^2}.$$

For any $z\in \R^m$ with positive coordinates, $g_\mu(z)>0$, and additionally

$$\det\left(\sum_{i=1}^m z_iv_iv_i^\intercal\right)>0.$$
Therefore, for every $t>0$, $t\bone\in\Ab_{p}$.

Now, by   \autoref{thm:jacobi},
\begin{eqnarray*}\Phi^i_p(y) &=& \frac{(\partial_i g_{\mu}(y))\cdot\det(\sum_{i=1}^m y_iv_iv_i^\intercal)}{g_\mu(y)\cdot\det(\sum_{i=1}^m y_i v_iv_i^\intercal)} 
+ \frac{g_\mu(y)\cdot(\partial_i\det(\sum_{i=1}^m y_iv_iv_i^\intercal))}
{g_\mu(y)\cdot\det(\sum_{i=1}^m y_iv_iv_i^\intercal)}\\
&=& \frac{\partial_i g_\mu(y)}{g_\mu(y)} + \trace\left(\left(\sum_{i=1}^m y_iv_iv_i^\intercal\right)^{-1} y_iy_i^\intercal\right)
\end{eqnarray*}
Therefore, since $g_\mu$ is homogeneous,
\begin{eqnarray*} 
\Phi^i_p(t\bone) &=& \frac{1}{t} \cdot \frac{\partial_i g_\mu(\bone)}{g_\mu(\bone)} +
\frac{\norm{v_i}^2}{t} \\
&=& \frac{\PP{S\sim\mu}{i\in S}}{t} + \frac{\norm{v_i}^2}{t} \leq \frac{\eps_1}{t} + \frac{\eps_2}{t} = \frac{\eps}{t}.
\end{eqnarray*}
The second identity uses \eqref{eq:marginaldef}.
Let $\phi=\eps/t$. Using $t=\delta$, it follows that
$$ \frac2{\delta}\phi + \phi^2 = \frac{2\eps}{t^2} + \frac{\eps^2}{t^2} = 1.$$
For $k\in[m]$ define
$$ p_k(y_1,\ldots,y_m) = \prod_{i=1}^k \left(1-\partial^2_{y_i}\right) \left(g_\mu(y)\cdot\det\left(\sum_{i=1}^m y_iv_iv_i^\intercal\right)\right),$$
and note that $p_m=Q$.
Let $x^0$ be the all-$t$ vector and $x^k$ be the vector that is $t+\delta$ in the first $k$ coordinates and $t$ in the rest.
By inductively applying \autoref{lem:rootsandbarrier}
and \autoref{lem:barriershift} for any $k\in[m]$,
$x^k$ is above all roots of $p_k$ and for all $i$,
$$\Phi^i_{p_k}(x_k) \leq \phi \Rightarrow \frac2{\delta}\Phi^i_{p_k}(x_i) + \Phi^i_{p_k}(x_i)^2 \leq 1.$$
Therefore, the largest root of $\mu[v_1,\ldots,v_m](x)$ is at most
$$ t+\delta = 2\sqrt{2\eps+\eps^2}.$$
\end{proofof}

\begin{proofof}{\autoref{thm:main}}
Let $\eps=\eps_1+\eps_2$ as always. 
\autoref{thm:barriermaxroot} implies that
the largest root of the mixed characteristic polynomial, $\mu[v_1,\ldots,v_m](x)$, is at most $2\sqrt{2\eps+\eps^2}.$
\autoref{thm:mixedinterlacing}
tells us that the polynomials $\{q_{S}\}_{S:\mu(S)>0}$
form an interlacing family.
So, by \autoref{thm:interlacingfamily} there is a set $S\subseteq [m]$ with $\mu(S)>0$ such that
the largest root of
$$ \det\left(x^2I-\sum_{i\in S} 2v_iv_i^\intercal\right)$$
is at most $2\sqrt{2\eps+\eps^2}$.
This implies that the largest root of
$$ \det\left(xI-\sum_{i\in S} 2v_iv_i^\intercal\right)$$
is at most $(2\sqrt{2\eps+\eps^2})^2$. Therefore,
$$ \norm{\sum_{i\in S} v_iv_i^\intercal} =\frac12\norm{\sum_{i\in S} 2v_iv_i^\intercal} \leq \frac12 (2\sqrt{2\eps+\eps^2})^2=4\eps+2\eps^2.$$
\end{proofof}

\section{Discussion}
Similar to  \cite{MSS13}
our main theorem is not algorithmic, i.e., we are not aware of any polynomial time algorithm that for a given homogeneous strongly Rayleigh distribution with small marginal probabilities and for a set of vectors assigned to the underlying elements with small norm finds a sample of the distribution with spectral norm bounded away from 1. Such an algorithm can lead to improved approximation algorithms for the Asymmetric Traveling Salesman Problem.

Although our main theorem can be seen as a generalization of \cite{MSS13} the bound that we prove on the maximum root of the mixed characteristic polynomial is incomparable to the bound of \autoref{thm:mss}. In \autoref{cor:KSr} we used our main theorem to prove Weaver's $\text{KS}_r$ conjecture \cite{Wea04} for $r>4$. It is an interesting question to see if the dependency on $\eps$ in our multivariate barrier can be improved, and if one can reprove  $\text{KS}_2$ using our machinery.

\subsubsection*{Acknowledgement}
We would like to thank Adam Marcus, Dan Spielman, and Nikhil Srivastava for stimulating discussions regarding the main obstacles in generalizing their proof of the Kadison-Singer problem. 

\bibliographystyle{alpha}
\bibliography{ref2}

\end{document}